\RequirePackage{fix-cm}
\documentclass[twocolumn]{svjour3}          
\smartqed  
\usepackage{amssymb}
\usepackage{graphicx}
\usepackage{amsmath}
\usepackage[sort&compress, square, numbers]{natbib}
\begin{document}

\title{Convergence Theorems for the Non-Local Means Filter
}


\author{Qiyu Jin \and Ion Grama \and Quansheng Liu     
}


\institute{Q. Jin \at
               UMR 7590, Institut de Min\'{e}ralogie et de Physique des Milieux Condens\'{e}s, Universit\'{e} Pierre et Marie Curie,  Campus Jussieu,
4 place Jussieu,
75005 Paris, France  \\
              Tel.: +33-144275241\\
              \email{Jin.Qiyu@impmc.upmc.fr}           
           \and
           I. Grama \at
               UMR 6205, Laboratoire de Mathématiques de Bretagne Atlantique, Universit\'{e} de Bretagne-Sud, Campus de Tohaninic, BP 573,
56017 Vannes, France
\\
Universit\'{e} Europ\'{e}ne de Bretagne, France
\\
              Tel.: +33-297017215\\
              \email{ion.grama@univ-ubs.fr}           
                    \and
           Q. Liu \at
               UMR 6205, Laboratoire de Mathématiques de Bretagne Atlantique, Universit\'{e} de Bretagne-Sud, Campus de Tohaninic, BP 573,
56017 Vannes, France
\\
Universit\'{e} Europ\'{e}ne de Bretagne, France
\\
              Tel.: +33-297017140\\
              \email{quansheng.liu@univ-ubs.fr}
}

\date{Received: date / Accepted: date}

\maketitle

\begin{abstract}
In this paper, we establish convergence theorems for the Non-Local Means Filter in removing the additive Gaussian noise.
We employ the techniques of "Oracle" estimation to determine the order of the widths of the similarity patches and search windows in the aforementioned filter.
We propose a practical choice of these parameters which improve the restoration quality of the filter compared with the usual choice of parameters.
\end{abstract}

\section{Introduction}

We deal with the additive
Gaussian noise model:%
\begin{equation}
Y(x)=f(x)+\varepsilon (x),\;x\in \mathbf{I,}  \label{s1Y1}
\end{equation}%
where $\mathbf{I}$ is the uniform $N\times N$ grid of pixels on the unit square, $Y=\left( Y\left( x\right) \right) _{x\in \mathbf{I}}$ is the observed
image brightness, $f:[0,1]^{2}\rightarrow \mathbf{R}_{+}$ is an original image (unknown target
regression function) and $\varepsilon =\left( \varepsilon \left( x\right)
\right) _{x\in \mathbf{I}}$ are independent and identically
distributed (i.i.d.) Gaussian random variables with mean $0$ and standard
deviation $\sigma >0.$

Important denoising techniques for the model (\ref{s1Y1})
have been developed in recent years. A very
significant step in these developments was the introduction of the Non-Local Means Filter by Buades et al \cite{buades2005review}.
For closely related works, see for example
\cite{polzehl2006propagation,kervrann2008local,buades2009note,Katkovnik2010local,lou2010image}.

The basic idea  of  the filters by weighted means is to estimate the unknown image $f(x_{0})$ by a weighted
average of observations $Y(x)$ of the form
\begin{equation}
\widetilde{f}_{w}(x_{0})=\sum_{x\in \mathbf{U_{x_0,h}}}w(x)Y(x),  \label{s1fx}
\end{equation}%
where for each $x_0$ and $h>0$, $\mathbf{U}_{x_0,h}$ denotes a square window with center $x_0$ and width $2h$,
 $w(x)$ are some
non-negative weights satisfying $\sum_{x\in \mathbf{U}_{x_0,h}}w(x)=1.$ The choice of
the weights $w(x)$ is usually based on two criteria: a spatial criterion so
that $w(x)$ is a decreasing function of the distance between $x$ and $x_0$,
and a similarity criterion so that $w(x)$ is also a decreasing function of the brightness difference $|Y(x)-Y(x_0)|$  (see e.g. \cite{yaroslavsky1985digital,tomasi1998bilateral}),
which measures the similarity between the pixels $x$ and $x_0$.
In the Non-Local Means Filter, $h>0$ can be chosen relatively large, and the weights $w(x)$ are calculated according to the similarity between data patches $\mathbf{Y}_{x,\eta}= (Y(y): y\in \mathbf{U}_{x,\eta})$
(identified as a vector whose composants are ordered lexicographically) and $\mathbf{Y}_{x_0,\eta}= (Y(y): y\in \mathbf{U}_{x_0,\eta})$,
 instead of the similarity between just the pixels $x$ and $x_0$.
  Here $\eta >0$ is the size parameter of data patches.

The  Non-Local Means Filter was further enhanced for speed in subsequent works by Mahmoudi
and Sapiro (2005 \citep{mahmoudi2005fast}), Bilcu and Vehvilainen (2007 %
\citep{bilcu2007fast}), Karnati, Uliyar and Dey(2009 %
\citep{karnati2009fast}), and Vignesh, Oh and Kuo (2010 %
\citep{vignesh2010fast}). Other authors as Kervrann and Boulanger (2006 %
\citep{kervrann2006optimal}, 2008 \citep{kervrann2008local}), Chatterjee and
Milanfar (2008 \citep{chatterjee2008generalization}), Buades, Coll and
Morel (2006 \citep{buades2006staircasing}), Dabov, Foi, Katkovnik and
Egiazarian (2007 \citep{dabov2007image}, 2009 \citep{dabov2009bm3d}) make
the Non-Local method better. Thacker, Bromiley and Manjn (2008 %
\citep{Thacker2008AQuantitative}) investigate this basis in order to
understand the conditions required for the use of Non-Local means, testing
the theory on simulated data and MR images of the normal brain. Katkovnik, Foi, Egiazarian and Astola (2010 \citep{Katkovnik2010local}) review
the evolution of the non-parametric regression modeling in imaging from the
local Nadaraya-Watson kernel estimate to the Non-Local means and further to
transform-domain faltering based on Non-Local block-matching.

Unfortunately, the ideal implementation of Non-Local Means is
computationally expensive. Therefore, for the sake of rising the speed of
denoising, only a neighborhood of the estimated point is considered. In
practice, the similarity patches of size $7\times 7$ or $9\times 9$ and
search windows of size $19\times 19$ or $21\times 21$ are often chosen.
However these choices are empirical and the problem of optimal choice
remains open. As a consequence, the results of the numerical simulations are
not always satisfactory.

In this paper, we use the statistic estimation and optimization techniques
to give a justification of the Non-Local Means filter, and to suggest the
order of sizes of search window and similarity patch.
Our main idea is to minimize a tight upper bound of the $L^{2}$ risk
\begin{equation*}
R\left( \widetilde{f}_{w}(x_{0})\right) =\mathbb{E}\left(
\widetilde{f}_{w}(x_{0})-f(x_{0})\right) ^{2}
\end{equation*}%
by changing the width of the search window. We first obtain an explicit formula for
the optimal weights $w^{\ast }_h$ in terms of the unknown function $f.$ The
corresponding weighted mean $f_h^{\ast }$ is called "Oracle";
the "Oracle" $f_h^{\ast }$ is shown to have an optimal rate of convergence and
high performance in numerical simulations. To mimic the "Oracle" $f_h^{\ast },$
we estimate $w_h^{\ast }$ by some adaptive weights $\widehat{w}_{h}$ based on
the  observed image $Y.$ We thus obtain the Non-Local Means Filter with
the proper width of window. Numerical results show that the Non-Local Means
Filter with proper width of window outperforms the Non-Local Means Filter
with standard choice.

The paper is organized as follows. In
Section \ref{sec_main results}, we introduce the "Oracle" estimator of Non-Local Means Filter and reconstruct Non-Local Means Filter with the idea of "Oracle" theory. Our main theoretical results are presented in Section \ref{sec theory nlm}
where we give the rate of convergence of the Non-Local Means Filter. In Section \ref{Sec:simulations}, we present our simulation results
with a brief analysis. Section \ref{sec conclusion} gives the conclusion of our paper.
Proofs of the main results are deferred to Section \ref{sec proofs of mains results}.

\section{Main results \label{sec_main results}}
\subsection{Notations}
Let us set some  notations to be used
throughout the paper. The Euclidean norm of a vector $x=\left(
x_{1},...,x_{d}\right) \in \mathbf{R}^{d}$ is denoted by $%
\left\Vert x\right\Vert _{2}=\left( \sum_{i=1}^{d}x_{i}^{2}\right)^{\frac{1}{2}} .$ The
supremum norm of $x$ is denoted by $\Vert x\Vert _{\infty }=\sup_{1\leq
i\leq d}\left\vert x_{i}\right\vert .$ The cardinality of a set $\mathbf{A}$ is denoted by $\text{card}\, \mathbf{A}$. For a positive integer $N$, the uniform $N\times N$ grid of pixels on the unit square is defined
by
\begin{equation}
\mathbf{I}=\left\{ \frac{1}{N},\frac{2}{N},\cdots ,\frac{N-1}{N},1\right\}
^{2}.  \label{def I}
\end{equation}%
Each element $x$ of the grid $\mathbf{I}$ will be called pixel. The number of pixels
is $n=N^{2}.$ For any pixel $x_{0}\in \mathbf{I}$ and a given $h>0,$ the
square window of pixels
\begin{equation}
\mathbf{U}_{x_{0},h}=\left\{ x\in \mathbf{I:\;}\Vert x-x_{0}\Vert _{\infty
}\leq h\right\}  \label{def search window}
\end{equation}%
will be called \emph{search window} at $x_{0}.$ We naturally take $h$ as a multiple of $\frac{1}{N}$ ($ h=\frac{k}{N}$ for some $k\in \{ 1, 2,\cdots,N\}$). The size of the square
search window $\mathbf{U}_{x_{0},h}$ is the positive integer number
\begin{equation}
M=(2Nh+1)^{2}=\mathrm{card\ }\mathbf{U}_{x_{0},h}. \label{defi M}
\end{equation}
 For any pixel $x\in \mathbf{U}%
_{x_{0},h}$ and a given $\eta >0$, a second square window of pixels
$
\mathbf{U}_{x,\eta }
$
will be called  \emph{ patch} at $x$. Like $h$, the parameter $\eta$ is also taken as a multiple of $\frac{1}{N}$. The size of the
 patch $\mathbf{U}_{x,\eta }$ is the positive integer
\begin{equation}
m=(2N\eta+1)^2=\mathrm{card\ }\mathbf{U}_{x_{0},\eta }.
\label{defi m}
\end{equation}
 The vector $\mathbf{Y}%
_{x,\eta }=\left( Y\left( y\right) \right) _{y\in \mathbf{U}_{x,\eta }}$
formed by the values of the observed noisy image $Y$ at pixels in the
patch $\mathbf{U}_{x,\eta }$ will be called simply \emph{data patch} at $%
x\in \mathbf{U}_{x_{0},h}.$

\subsection{the "Oracle" of Non-Local means}

In order to study statistic estimation theory of the Non-Local Means
algorithm, we introduce an "Oracle" estimator (for details on this concept see
Donoho and Johnstone (1994 \citep{donoho1994ideal})) of Non-Local means.
Denote
\begin{equation}
f^{\ast}_{h}=\sum_{x\in \mathbf{U}_{x_0,h}} w^{\ast}_{h}Y(x),
\label{oracle_extimate}
\end{equation}
where
\begin{equation}
w_h^*(x)=e^{-\frac{\rho_{f,x_0}^2(x)}{H^2}} \bigg/\sum_{y\in \mathbf{U}_{x_0,h}}e^{-\frac{\rho_{f,x_0}^2(x)}{H^2}},  \quad x \in \mathbf{U}_{x_0,h},
\label{oracle_weights}
\end{equation}
\begin{equation}
\rho_{f,x_0}(x)\equiv|f(x)-f(x_0)|
\label{simulation function}
\end{equation}
  and $H>0$ is a constant. It is obvious that
\begin{equation}
\sum_{x\in \mathbf{U}_{x_0,h}} w_h^*(x)=1 \quad \text{and} \quad w_h^*(x)\geq 0.
\label{condition weights}
\end{equation}
Note that the function $\rho_{f,x_0}(x) \geq 0$ characterizes the similarity of the image brightness
at the pixel $x$ with respect to the pixel $x_{0}$, therefore we shall call $%
\rho_{f,x_0} $ similarity function. The usual bias-variance decomposition of the
Mean Squared Error (MSE)%

\begin{equation}
\begin{split}
&\mathbb{E}\left( f(x_{0})-f^*_{h}(x_{0})\right) ^{2}
\\&=
\left( \sum_{x\in
\mathbf{U}_{x_0,h}}w_h^*(x)\left( f(x)-f(x_{0})\right) \right) ^{2}+\sigma
^{2}\sum_{x\in \mathbf{U}_{x_0,h}}w_h^*(x)^{2}
\\&\leq
\left( \sum_{x\in
\mathbf{U}_{x_0,h}}w_h^*(x)\left| f(x)-f(x_{0})\right| \right) ^{2}+\sigma
^{2}\sum_{x\in \mathbf{U}_{x_0,h}}w_h^*(x)^{2}
\label{MSE}
\end{split}
\end{equation}%
The inequality (\ref{MSE}) combining with (\ref{simulation function}) implies the following upper bound
\begin{equation}
\mathbb{E}\left( f(x_{0})-f^*_{h}(x_{0})\right) ^{2}\leq g(w_h^*),
\label{upper_bound}
\end{equation}%
where
\begin{equation}
g(w)=\left( \sum_{x\in \mathbf{U}_{x_0,h}}w(x)\rho_{f,x_0} (x)\right) ^{2}+\sigma
^{2}\sum_{x\in \mathbf{U}_{x_0,h}}w(x)^{2}.  \label{function upper bound}
\end{equation}
We shall define a family of estimates by minimizing the function $g\left(
w^*_h\right) $ in $w^*_h$ and plugging the optimal weights into (\ref%
{oracle_extimate}). We shall consider the local H\"{o}lder condition
\begin{equation}
|f(x)-f(y)|\leq L\Vert x-y\Vert _{\infty }^{\beta },\,\,\,\forall x,\,y\in
\mathbf{U}_{x_{0},h+\eta},  \label{Local Lip cond}
\end{equation}%
where $\beta >0$ and $L>0$ are constants, $h>0$, $\eta>0$ and $x_{0}\in \mathbf{I}.$ The following theorem
gives the rate of convergence of the "Oracle" estimator and the proper width $h$
of the search window.

\begin{theorem}
\label{th_oracle} Assume  that $h=\left(\frac{%
\sigma^2}{4\beta L^2}\right)^{\frac{1}{2\beta+2}}n^{-\frac{1}{2\beta+2}}$
and $H > \sqrt 2 Lh^{\beta}$. Suppose that the function $f$ satisfies the local H\"{o}%
lder condition (\ref{Local Lip cond}) and  $f^*_h(x_0)$ is given by (\ref%
{oracle_extimate}). Then
\begin{equation}
\mathbb{E}\left( f^*_h(x_0)-f(x_0)\right)^2\leq \frac{2^{\frac{2\beta + 6}{2\beta +2}%
} \sigma^{\frac{4\beta}{2\beta +2}} L^{\frac{4}{2\beta +2}} }{\beta^{\frac{%
2\beta}{2\beta +2}}} n^{-\frac{2\beta}{2\beta+2}}.  \label{rate_oracle}
\end{equation}
\end{theorem}

For the proof of this theorem see Section \ref{proof of th oracle}.

We confirm the theorem by simulations that the difference between the "Oracle" $%
f_{h}^{\ast }(x_{0})$ and the true value $f(x_{0})$ is extremely small (see
Table\ \ref{Tab oracle} and the definition of PSNR can be found in Section
\ref{Sec:simulations} ). The latter, at least from the practical point of
view, the theorem justifies that it is reasonable to optimize the upper bound $g(w_h^*)$
instead of optimizing the risk $\mathbb{E}\left( f_{h}^{\ast }(x_{0})-f(x_{0})\right)
^{2}$ itself.

Theorem \ref{th_oracle} displays  that the choice of a small search window,
in the place of the whole observed image, suffices to ensure a denoising
without loss of visual quality, and explains why we take a small search
window for the simulations in the Non-Local Means algorithm.

\begin{table*}[tbp]
\caption{
PSNR values when "Oracle" estimator $f_{h}^{\ast }$ is applied with
different values of $M$. }
\label{Tab oracle}
\begin{center}
\renewcommand{\arraystretch}{0.6} \vskip3mm {\fontsize{8pt}{\baselineskip}%
\selectfont
\begin{tabular}{|c|rrrrr|}
\hline
Image & Lena & Barbara & Boats & House & Peppers \\
Size & $512\times512$ & $512\times512$ & $512\times512$ & $256\times256$ & $%
256\times256$ \\ \hline\hline
$\sigma /PSNR$ & 10/28.12db & 10/28.12db & 10/28.12db & 10/28.11db &
10/28.11db \\ \hline
$9 \times 9$ & 38.98db & 37.26db & 37.66db & 38.93db & 37.85db \\
$11 \times 11$ & 40.12db & 38.49db & 38.80db & 40.04db & 38.85db \\
$13 \times 13$ & 41.09db & 39.55db & 39.78db & 40.98db & 39.64db \\
$15 \times 15$ & 41.92db & 40.45db & 40.63db & 41.77db & 40.39db \\
$17 \times 17$ & 42.64db & 41.23db & 41.39db & 42.40db & 41.00db \\
$19 \times 19$ & 43.29db & 41.93db & 42.06db & 43.06db & 41.58db \\
$21 \times 21$ & 43.88db & 42.57db & 42.67db & 43.61db & 42.14db \\
\hline\hline
$\sigma /PSNR$ & 20/22.11db & 20/22.11db & 20/22.11db & 20/28.12db &
20/28.12db \\ \hline
$9 \times 9 $ & 33.61db & 31.91db & 32.32db & 33.72db & 32.62db \\
$11 \times 11$ & 34.78db & 33.20db & 33.49db & 34.92db & 33.65db \\
$13 \times 13$ & 35.80db & 34.28db & 34.49db & 35.98db & 34.51db \\
$15 \times 15$ & 36.69db & 35.22db & 35.40db & 36.80db & 35.26db \\
$17 \times 17$ & 37.48db & 36.05db & 36.20db & 37.48db & 35.89db \\
$19 \times 19$ & 38.17db & 36.74db & 36.90db & 38.07db & 36.45db \\
$21 \times 21$ & 38.80db & 37.40db & 37.54db & 38.67db & 36.98db \\
\hline\hline
$\sigma /PSNR$ & 20/22.11db & 20/22.11db & 20/22.11db & 20/28.12db &
20/28.12db \\ \hline
$9 \times 9$ & 30.65db & 28.89db & 29.25db & 30.69db & 29.51db \\
$11 \times 11$ & 31.83db & 30.23db & 30.45db & 31.90db & 30.51db \\
$13 \times 13$ & 32.85db & 31.33db & 31.49db & 32.92db & 31.34db \\
$15 \times 15$ & 33.74db & 32.27db & 32.37db & 33.76db & 32.08db \\
$17 \times 17$ & 34.50db & 33.09db & 33.16db & 34.48db & 32.74db \\
$19 \times 19$ & 35.20db & 33.81db & 33.85db & 35.13db & 33.32db \\
$21 \times 21$ & 35.79db & 34.46db & 34.48db & 35.71db & 33.85db \\ \hline
\end{tabular}
} \vskip1mm
\end{center}
\end{table*}

\subsection{Reconstruction of Non-Local Means filter}
With the theory of "Oracle" estimator, we reconstruct the Non-Local Means filter \citep{buades2005review}. Let $h>0$ and $\eta >0$ be fixed numbers. For any $x_{0}\in \mathbf{I}$ and
any $x\in \mathbf{U}_{x_{0},h}$,  the distance between the data patches
$\mathbf{Y}_{x,\eta }=\left( Y\left( y\right) \right) _{y\in \mathbf{U}%
_{x,\eta }}$ and $\mathbf{Y}_{x_{0},\eta }=\left( Y\left( y\right) \right)
_{y\in \mathbf{U}_{x_{0},\eta }}$ is defined by
\begin{equation*}
d^{2}\left( \mathbf{Y}_{x,\eta },\mathbf{Y}_{x_{0},\eta }\right) =
\left\Vert \mathbf{Y}_{x,\eta }-\mathbf{Y}_{x_{0},\eta }\right\Vert _{2}^{2},
\end{equation*}%
where  $\left\Vert \mathbf{Y}_{x,\eta }-\mathbf{Y}_{x_{0},\eta }\right\Vert _{2}^{2}=\frac{1}{m}\sum\limits_{y\in \mathbf{U}_{x,\eta }} $ $\left( Y(T_xy)-Y(y)\right)^2$, $T_x$ is the translation mapping: $T_x y=x+(y-x_{0})$ and $m$ is given by (\ref{defi m}),
which measures the similarity between the data patches $\mathbf{Y}_{x,\eta}$ and $\mathbf{Y}_{x_0,\eta}$.
Since $%
|f(x)-f(x_{0})|^{2}=\mathbb{E}|Y(x)-Y(x_{0})|^{2}-2\sigma ^{2}$, an obvious
estimator of $\mathbb{E}\left\vert Y(x)-Y(x_{0})\right\vert ^{2}$ is given by
$
d^{2}\left( \mathbf{Y}_{x,\eta },\mathbf{Y}_{x_{0},\eta }\right)
$.
Define an estimated similarity function $\widehat{\rho }_{x_0}$ by
\begin{equation}
\widehat{\rho }_{x_0}^{2}(x)=d^{2}\left( \mathbf{Y}_{x,\eta },\mathbf{Y}_{x_{0},\eta }\right)-2\sigma ^{2},  \label{estimator similar}
\end{equation}%
and an adaptive estimator $\widehat{f}_{h }$ by
\begin{equation}
\widehat{f}_{h }(x_{0})=\sum_{x\in  \mathbf{U}_{x_{0},h}}\widehat{w}%
_{h}(x)Y(x),  \label{estimate}
\end{equation}%
where

\begin{eqnarray}
\widehat{w}_{h}
&=&
e^{-\frac{\widehat{\rho }_{x_0}^{2}(x)}{H^{2}}}\bigg/\sum_{x^{\prime }\in
\mathbf{U}_{x_{0},h}}e^{-\frac{\widehat{\rho }_{x_0}^{2}(x^{\prime })}{%
H^{2}}}
\nonumber\\&=&
e^{-\frac{d^{2}\left( \mathbf{Y}_{x,\eta },\mathbf{Y}_{x_{0},\eta }\right)}{H^{2}}}\bigg/\sum_{x^{\prime }\in
\mathbf{U}_{x_{0},h}}e^{-\frac{d^{2}\left( \mathbf{Y}_{x',\eta },\mathbf{Y}_{x_{0},\eta }\right)}{%
H^{2}}}.
\label{estimate_weights}
\end{eqnarray}
 and
  $\mathbf{U}_{x_{0},h}$ given by (\ref{def search window}).

Note that $\Delta _{x_{0},x}(y)=|f(y)-f(T-xy)|$ and $\zeta(y)=\epsilon(T-xy)-\epsilon(y)$. It is easy to see
that

\begin{eqnarray*}
\widehat{\rho }^2_{x_0}(x)
&=&
 \frac{1}{m}\sum\limits_{y\in \mathbf{U}_{x_{0},\eta }} \left( f(x)+\epsilon(x) - f(x_0)-\epsilon(x_0)\right)^2 - 2\sigma^2
\\
   &\leq&
    \frac{1}{m}\sum\limits_{y\in \mathbf{U}_{x_{0},\eta }} \left( \Delta _{x_{0},x}(y)+\zeta(y)\right)^2 - 2\sigma^2
   \\&=&
   \frac{1}{m}\sum\limits_{y\in \mathbf{U}_{x_{0},\eta }}\Delta^2 _{x_{0},x}(y)
   \\&&
   +\frac{1}{m}\sum\limits_{y\in \mathbf{U}_{x_{0},\eta }} \left(\zeta \left( y\right) ^{2}-2\sigma
^{2}+2\Delta_{x_0,x} \left( y\right) \zeta \left( y\right)  \right)
\\&=&
    \frac{1}{m}\sum\limits_{y\in \mathbf{U}_{x_{0},\eta }}\Delta^2 _{x_{0},x}(y) +\frac{1}{m}S(x),
\end{eqnarray*}
where
\begin{equation}
S(x)=\sum\limits_{y\in \mathbf{U}_{x_{0},\eta }} \left(\zeta \left( y\right) ^{2}-2\sigma
^{2}+2\Delta_{x_0,x} \left( y\right) \zeta \left( y\right)  \right).
\label{defi sx}
\end{equation}

\begin{theorem}
\label{th similar function} Assume that $\eta =c_{0}n^{-\alpha }$ $\left( \frac{(1-\beta )^{+}%
}{2\beta +2}<\alpha <\frac{1}{2}\right) $. Suppose that the function $f$ satisfies the
local H\"{o}lder conditions (\ref{Local Lip cond}) and $\widehat{\rho }_{x_0}$ is given by (\ref{estimator similar}). Then there is a constant $c_{1}$
such that
\begin{eqnarray}
&&\mathbb{P}\left( \max_{x\in \mathbf{U}_{x_{0},h}}\left\vert \widehat{\rho }_{x_0}%
^{2}(x)-\rho_{f,x_0} ^{2}(x)\right\vert \geq c_{1}n^{\alpha -\frac{1}{2}}\sqrt{\ln n}%
\right)
\nonumber\\&=&O\left( n^{-1}\right) .  \label{rate similar function}
\end{eqnarray}
\end{theorem}

For the proof of this theorem see Section \ref{proof of similar condition}.

In the Theorem  \ref{th similar
function}, we consider that $\sigma$ is a constant and the H\"{o}lder condition (\ref{Local Lip cond}) implies that
$\left\vert \frac{1}{\mathrm{card} \mathbf{U}''_{x_{0},\eta }}\right.$
$\left.\Delta _{x_{0},x}^{2}(y)-\rho_{f,x_0}
^{2}(x)\right\vert =O\left(n^{\frac{2\beta}{2\beta+2}}\right)$. Therefore, if $n $ is large enough,
 we have $%
\left\vert \frac{1}{\mathrm{card} \mathbf{U}''_{x_{0},\eta }}\Delta _{x_{0},x}^{2}(y)-\rho_{f,x_0}
^{2}(x)\right\vert \ll \sigma $. It is to say that the larger
the standard deviation of the noise is, the more useful our theorem will be.
We take the test image "Lena" as an example, which is degraded by Gaussian noise with $\sigma=10$, $\sigma=20$ and $\sigma=30$ respectively. We fix the size of search window $M=13\times13$, $H=0.4\times \sigma +2$ and choose the size of similarity patch $m\in \{2k+1: k=1,2,\cdots,20\}$. In the cases of $\sigma=20$ and $\sigma=30$,
Figure \ref{Fig evolution} (b) and (c) illustrate  that  the value of PSNR value of increases when the  size of a similarity patch
increases. The evolutions of PSNR value are in
accordance with Theorem \ref{th similar function}.
However, in the case of $\sigma=10$, Figure  \ref{Fig evolution}  (a) displays that the PSNR value increases  when the size of a similarity patch
increases in the interval $[3,15]$ and reaches the peak value. But it decreases  in
the interval $[15,41]$.  This means that  the value $\sigma =10$ is
not large enough to satisfy the condition  $\left\vert \frac{1}{\mathrm{card} \mathbf{U}''_{x_{0},\eta }}\Delta _{x_{0},x}^{2}(y)-\rho_{f,x_0}
^{2}(x)\right\vert \ll \sigma $.

\begin{figure*}[tbp]
\begin{center}
\renewcommand{\arraystretch}{0.2} \addtolength{\tabcolsep}{-6pt} \vskip3mm {%
\fontsize{8pt}{\baselineskip}\selectfont
\begin{tabular}{ccc}
\includegraphics[width=0.33\linewidth]{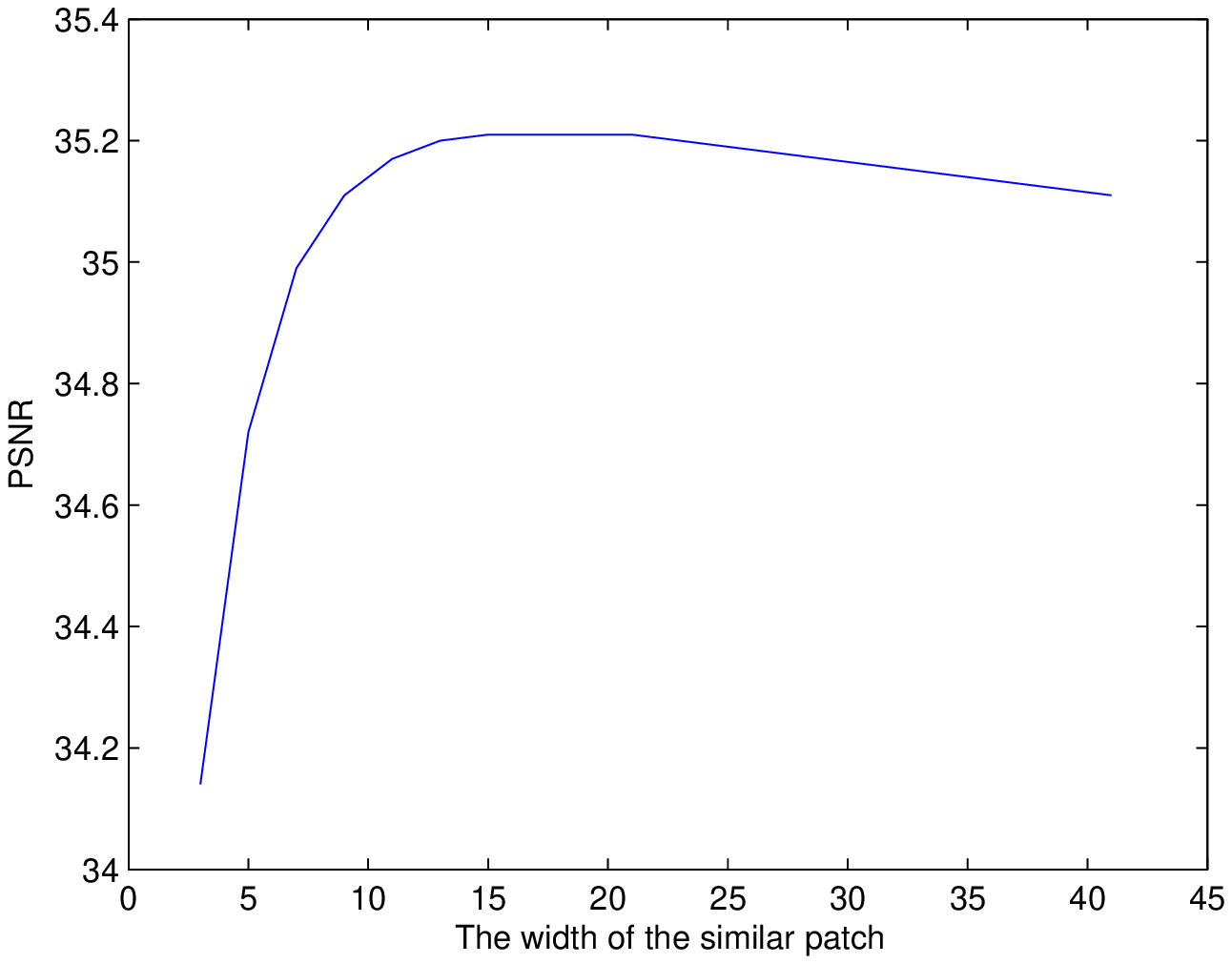} & %
\includegraphics[width=0.33\linewidth]{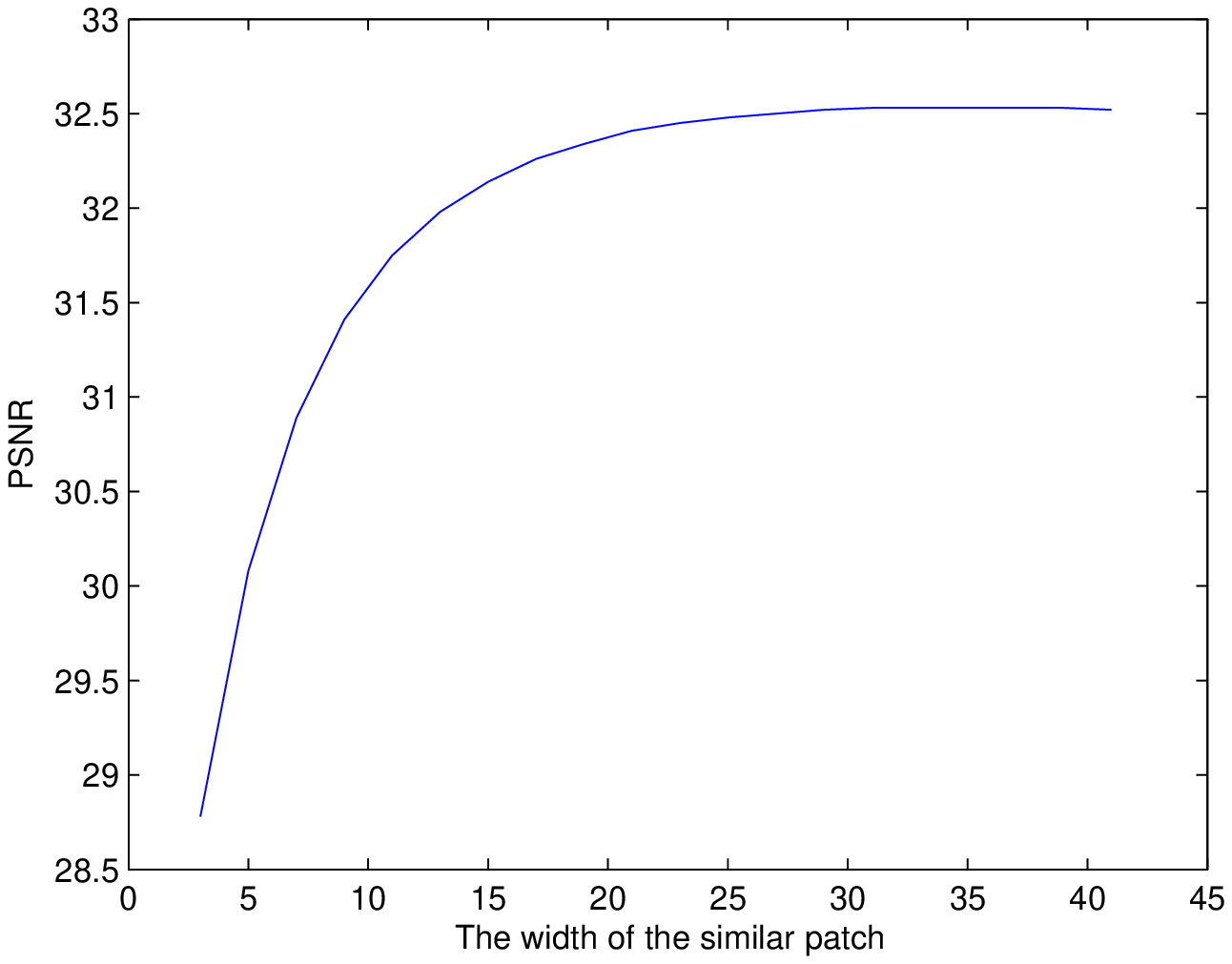} & %
\includegraphics[width=0.33\linewidth]{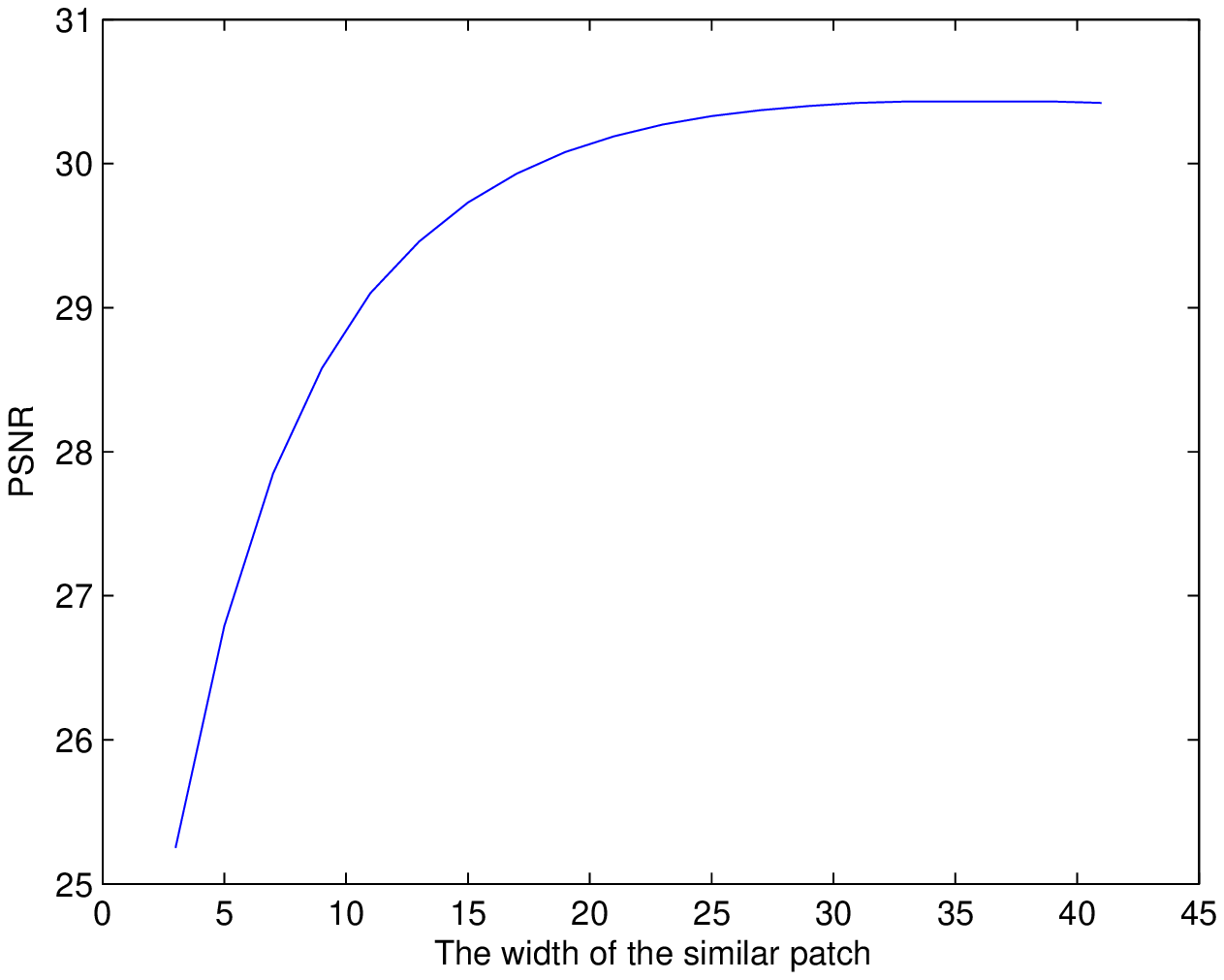} \\
(a)$\sigma=10$&(b)$\sigma=20$&(c)$\sigma=30$%
\end{tabular}
} \vskip1mm
\end{center}
\par
\rule{0pt}{-0.2pt}%
\caption{The evolution of PSNR value as a function of  the size of a similarity patch.}
\label{Fig evolution}
\end{figure*}

In order to improve the results, we sometimes shall use the smoothed version of the
estimate of brightness variation $d_\kappa^{2}\left( \mathbf{Y}_{x,\eta },\mathbf{Y}_{x_{0},\eta }\right)$ instead of the non smoothed one $d^{2}\left( \mathbf{Y}_{x,\eta },\mathbf{Y}_{x_{0},\eta }\right)$. It should be noted that for
the smoothed versions of the estimated brightness variation we can establish
similar convergence results. The smoothed estimator $d_\kappa^{2}\left( \mathbf{Y}_{x,\eta },\mathbf{Y}_{x_{0},\eta }\right)=\left\Vert \mathbf{Y}_{x,\eta }-\mathbf{Y}_{x_{0},\eta }\right\Vert _{2,K}^{2}$ is defined by%
\begin{equation}
\left\Vert \mathbf{Y}_{x,\eta }-\mathbf{Y}_{x_{0},\eta }\right\Vert _{2,\kappa}^{2}=\frac{\sum\limits_{y\in \mathbf{U}_{x_{0},\eta }}\kappa(y)\left( Y(T-xy)-Y(y)\right)^2}{\sum\limits_{y\in \mathbf{U}_{x_{0},\eta }}\kappa(y)}
\label{distance}
\end{equation}%
where $\kappa(y)$ are some weights defined on $\mathbf{U}_{x_{0},\eta }.$
With the rectangular kernel%
\begin{equation}
\kappa_{r}\left( y\right) =\left\{
\begin{array}{ll}
1, & y\in \mathbf{U}''_{x_{0},\eta }, \\
0, & \text{otherwise,}%
\end{array}%
\right.   \label{rect kernel}
\end{equation}%
we obtain exactly the distance $d^{2}\left( \mathbf{Y}_{x,\eta },\mathbf{Y}_{x_{0},\eta }\right) $. Other smoothing kernels $\kappa(y)$ used in the simulations are the Gaussian
kernel
\begin{equation}
\kappa_{g}(y)=\exp \left( -\frac{N^{2}\Vert y-x_{0}\Vert _{2}^{2}}{2h_g}\right) ,
\label{s4kg}
\end{equation}%
where $h_g$ is the bandwidth parameter, and the following kernel: for $y \in \mathbf{U}_{x_0,\eta}$,
\begin{equation}
\kappa_{0}\left( y\right) =\sum_{k=\max(1,j)}^{N\eta}\frac{1}{(2k+1)^2}
\end{equation}%
if $\|y-x_0\|_{\infty}=\frac{j}{N}$ for some $j\in \{0,1,\cdots,N\eta\}$.
 $\kappa(y)=\kappa_{0}(y)$ is used in our paper and Buades et al \citep{buades2005review}.

To avoid the undesirable border effects, we mirror the image outside the
image limits.
In more detail, we extend the image outside the image limits
symmetrically with respect to the border. At the corners, the image is
extended symmetrically with respect to the corner pixels.

The following is the algorithm for denoising used in Buades et al\citep{buades2005review}.

\noindent\rule{.48\textwidth}{.2pt}

\textbf{Algorithm}\quad NL-means (Buades et al \citep{buades2005review},\\
http://dmi.uib.es/abuades/nlmeanscode.html).

\noindent\rule{.48\textwidth}{.2pt}

Let $\{H,M,m\}$ be the parameters.

Repeat for each $x_0\in \mathbf{I}$

\quad - compute

\quad \quad $d_\kappa^{2}\left( \mathbf{Y}_{x,\eta },\mathbf{Y}_{x_{0},\eta }\right)$ (given by (\ref{distance}))

\quad \quad and $d_\kappa^{2}\left( \mathbf{Y}_{x_{0},\eta },\mathbf{Y}_{x_{0},\eta }\right)=\max\{ d_\kappa^{2}\left( \mathbf{Y}_{x,\eta },\mathbf{Y}_{x_{0},\eta }\right): x\neq x_0, x\in
\mathbf{U}_{x_0,h}\}$

\quad \quad $w(x)=\frac{\exp(-d_\kappa^{2}\left( \mathbf{Y}_{x,\eta },\mathbf{Y}_{x_{0},\eta }\right)/H^{2})}{\sum_{y\in \mathbf{U}_{x_0,h}} \exp(-d_\kappa^{2}\left( \mathbf{Y}_{y,\eta },\mathbf{Y}_{x_{0},\eta }\right)/H^{2})}$ (see the equation (\ref{estimate_weights}))

\quad \quad $\widehat{f}(x_0)=\sum_{x\in \mathbf{U}_{x_0,h}}w(x)Y(x)$

\noindent\rule{.48\textwidth}{.2pt}

 A
detailed theory analysis and the convergence of Non-Local Means Filter will be given in Section \ref{sec theory nlm}.  In Section \ref{Sec:simulations},
the numerical simulations show that we can optimize the parameters to make Non-Local Means Filter better.

\section{Convergence theorem of Non-Local means\label{sec theory nlm} }
Now, we turn to the study of the convergence of the Optimal Weights Filter. Due to the difficulty in dealing with the dependence of the weights we shall consider a slightly modified version of the proposed algorithm:  we divide  the set of pixels  into two independent parts,  so that the weights are constructed from the one part,  and  the estimation of the target function is a weighted mean  along  the other part. More precisely,  assume  that $x_{0}\in \mathbf{I},$ $h>0$ and $\eta >0.$ To prove the
convergence we split the set of pixels into two parts $\mathbf{I}=\mathbf{I}'_{x_{0}}\cup \mathbf{I}''_{x_{0}},$ where
\begin{equation*}
\mathbf{I}'_{x_{0}}=\left\{ x_{0}+\left( \frac{i}{N},\frac{j}{N}\right) \in
\mathbf{I}:i+j\text{ is pair }\right\} ,
\end{equation*}%
and $\mathbf{I}''_{x_{0}}=\mathbf{I}\diagdown \mathbf{I}'_{x_{0}}.$
Define an estimated similarity function $\widehat{\rho }'_{x_0}$ is given by
\begin{equation}
\widehat{\rho }_{x_0}^{'2}(x)=\frac{1}{\mathrm{card} \mathbf{U}''%
_{x_{0},\eta }}\sum_{y\in {\mathbf{U}''_{x_{0},\eta }}}|Y(y)-Y(T-xy)|^{2}-2%
\sigma ^{2},  \label{estimator similar}
\end{equation}%
where $\mathbf{U}_{x_{0},\eta}^{'' }=\mathbf{U}_{x_{0},h}\cap \mathbf{I}''%
_{x_{0}}$ with $\mathbf{U}_{x_{0},h}$ given by (\ref{def search window}).
Then an adaptive estimator $\widehat{f}'_{h }$ by
\begin{equation}
\widehat{f}'_{h }(x_{0})=\sum_{x\in  \mathbf{I}_{x_{0},1}}\widehat{w}%
_{h}(x)Y(x),  \label{estimate}
\end{equation}%
where \begin{equation}
\widehat{w}'_{h}=e^{-\frac{\widehat{\rho }_{x_0}^{'2}(x)}{H^{2}}}\bigg/\sum_{x^{\prime }\in
\mathbf{U}_{x_{0},h}^{\prime }}e^{-\frac{\widehat{\rho }_{x_0}^{'2}(x^{\prime })}{%
H^{2}}}.  \label{estimate_weights}
\end{equation}
 and
 $\mathbf{U}_{x_{0},h}^{\prime }=\mathbf{U}_{x_{0},h}\cap \mathbf{I}'_{x_{0}}$ with $\mathbf{U}_{x_{0},h}$ given by (\ref{def search window}).

In the next theorem we prove that the Mean Squared Error of the estimator $%
\widehat{f}'_{h}(x_{0})$ converges at the rate $n^{-\frac{2\beta }{2\beta +2}%
} $ which is the usual optimal rate of convergence for a given H\"{o}lder
smoothness $\beta >0$ (see e.g. Fan and Gijbels (1996 \citep{FanGijbels1996}%
)).

\begin{theorem}
\label{th rate estimator} Let $\eta=c_0n^{-\alpha}$, $h=\left(%
\frac{\sigma^2}{4\beta L^2}\right)^{\frac{1}{2\beta+2}}n^{-\frac{1}{2\beta+2}%
}$ and $H>2c_1n^{\alpha-\frac{1}{2}}$ and $H>\sqrt{2}Lh $. Suppose that the function f satisfies the H\"{o}%
lder condition (\ref{Local Lip cond}) and $\widehat{f}'_h$ is given by (\ref{estimate}).
Then
\begin{eqnarray*}
&&\mathbb{E} \left(\widehat{f}'_{h}(x_0)-f(x_0)\right) ^2
\\&\leq &2\left(\frac{2^{%
\frac{2\beta + 6}{2\beta +2}} \sigma^{\frac{4\beta}{2\beta +2}} L^{\frac{4}{%
2\beta +2}} }{\beta^{\frac{2\beta}{2\beta +2}}} \right)\left( \frac{1+\frac{%
2c_1n^{\alpha-\frac{1}{2}}}{H}}{1-\frac{c_1n^{\alpha-\frac{1}{2}}}{H}}%
\right) ^2n^{-\frac{2\beta}{2\beta+2}}.  \label{rate etimator}
\end{eqnarray*}
\end{theorem}

For the proof of this theorem see Section \ref{proof of similar condition}.

\begin{table*}[tbp]
\caption{
Comparison between the Non-Local Means Filter with Baude's
parameters and our parameters. }
\label{tab psnr nlmeans}
\begin{center}
\renewcommand{\arraystretch}{0.6} \vskip3mm {\fontsize{8pt}{\baselineskip}%
\selectfont
\begin{tabular}{|l|rrrrr|}
\hline
Image & Lena & Barbara & Boats & House & Peppers \\
Size & $512\times512$ & $512\times512$ & $512\times512$ & $256\times256$ & $%
256\times256$ \\ \hline\hline
$\sigma/PSNR$ & 10/28.12db & 10/28.12db & 10/28.12db & 10/28.11db &
10/28.11db \\ \hline
PSNR/Buade & 34.99db & 33.82db & 32.85db & 35.50db & 33.13db \\
PSNR/Ours & 35.22db & 33.55db & 33.00db & 35.35db & 33.16db \\
$\Delta$PSNR & 0.23db & -0.27db & 0.15db & -0.15db & 0.03db \\ \hline\hline
$\sigma/PSNR$ & 20/22.11db & 20/22.11db & 20/22.11db & 20/28.12db &
20/28.12db \\ \hline
PSNR/Buade &  31.51db & 30.38db & 29.32db & 32.51db & 29.73db \\
PSNR/Ours & 32.39db & 30.62db & 30.02db & 32.57db & 30.30db \\
$\Delta$PSNR & 0.82db & 0.24db & 0.70db & 0.08db & 0.57db \\ \hline\hline
$\sigma/PSNR$ & 30/18.60db & 30/18.60db & 30/18.60db & 30/18.61db &
30/18.61db \\ \hline
PSNR/Buade & 28.86db & 27.65db & 27.38db & 29.17db & 27.67db \\
PSNR/Ours & 30.20db & 28.06db & 28.60db & 30.49db & 28.28db \\
$\Delta$PSNR & 1.34db & 0.41db & 1.22db & 1.32db & 0.61db \\ \hline
\end{tabular}
} \vskip1mm
\end{center}
\end{table*}

\begin{table*}[tbp]
\caption{
Performance of denoising algorithms when applied to test noisy (WGN) images. }
\label{Table compar}
\begin{center}
\renewcommand{\arraystretch}{0.6} \vskip3mm {\fontsize{8pt}{\baselineskip}%
\selectfont
\begin{tabular}{|c|l|ccccc|}
\hline
            &Images & Lena & Barbara & Boat & House & Peppers \\
            &Sizes & $512 \times 512$ & $512 \times 512$ & $512 \times 512$ & $256 \times
            256$ & $256 \times 256$ \\ \hline\hline
$\sigma$    &         Method  & PSNR    &  PSNR   &  PSNR   &  PSNR   &  PSNR   \\ \hline
                       &Non-Local Means &&&&&\\
                        &$M=13\times13$
            & 32.39db & 30.62db & 30.02db & 32.57db & 30.30db \\
                        &$m=21\times21$&&&&&\\ \cline{2-7}
                        &Buades et al\citep{Bu}
            & 31.51db & 30.38db & 29.32db & 32.51db & 29.73db \\
                        &Salmon et al \citep{salmon2009nl}
            & -       &      -  & -       & -       & 29.46db \\
                        &Katkovnik et al \citep{katkovnik2004directional}
            & 30.74db & 27.38db & 29.03db & 31.24db & 29.58db \\
$20$                    &Foi et al \citep{foinovel}
            & 31.43db & 27.90db & 39.61db & 31.84db & 30.30db \\
                        &Roth et al \citep{roth2009fields}
            & 31.89db & 28.28db & 29.86db & 32.29db & 30.47db \\
                        &Hirkawa et al \citep{hirakawa2006image}
            & 32.69db & 31.06db & 30.25db & 32.58db & 30.21db \\
                        &Kervrann et al \citep{kervrann2008local}
            & 32.64db & 30.37db & 30.12db & 32.90db & 30.59db \\
                        &Jin et al \citep{jin2011removing}
            & 32.68db & 31.04db & 30.30db & 32.83db & 30.61db \\
                        &Hammond et al \citep{hammond2008image}
            & 32.81db & 30.76db & 30.41db & 32.52db & 30.40db \\
                        &Aharon et al \citep{aharon2006rm}
            & 32.39db & 30.84db & 30.39db & 33.10db & 30.80db \\
                        &Dabov et al \citep{dabov2007image}
            & 33.05db & 31.78db & 30.88db & 33.77db & 31.29db \\
            \hline
\end{tabular}
} \vskip1mm
\end{center}
\end{table*}

\section{\label{Sec:simulations}Simulation }

In this section, we compare the performance of the Non-Local Means Filter
computed using the parameters proposed in this paper with those proposed in
Buades et al  \citep{buades2005review}. The results were measured by the usual Peak
Signal-to-Noise Ratio (PSNR) in decibels (db) defined as%
\begin{equation*}
PSNR=10\log _{10}\frac{255^{2}}{MSE},
\end{equation*}
\begin{equation*}
MSE=\frac{1}{\mathrm{card} \mathbf{I}}%
\sum\limits_{x\in \mathbf{I}}(f(x)-\widehat{f}_h(x))^{2},
\end{equation*}%
where $f$ is the original image and $\widehat{f}_h$ is the estimated one.

We have done simulations on a commonly-used set of images available at
http://decsai.\newline
ugr.es/javier/denoise/test images/. The potential of the estimation method
is illustrated with the $512\times 512$ "Lena" image (Figure\ \ref{fig4}(a))
corrupted by an additive white Gaussian noise (Figure\ \ref{fig4}(a) right,
PSNR$=22.10db$, $\sigma =20$). We have seen experimentally that the
filtering parameter $H$ can take values between $0.4\times \sigma +2$ and $%
0.5\times \sigma +2$, obtaining a high visual quality solution. Theorem \ref%
{th rate estimator} implies that the search window is of size $c_{0}\sigma ^{%
\frac{2}{2\beta +2}}$. Assuming that $\beta =1$, we get a search window of
size $c_{0}\sqrt{\sigma }$. Experimentations show that when the size of the
search window takes values $1.5\times \sqrt{\sigma }+4.5$, we obtain  the
best quality for Non- Local Means Filter. Our simulations also show that it
is convenient to take the similarity patch size as $m=17\times 17$ for $%
\sigma =10,$ and $m=21\times 21$ for $\sigma =20$ and $\sigma =30$. In
Figure\ \ref{fig4}(b) left, we can see that the noise is reduced in a
natural manner and significant geometric features, fine textures, and
original contrasts are visually well recovered with no undesirable artifacts
(PSNR$=32.39db$). To better appreciate the accuracy of the restoration
process, the square of difference between the original image and the
recovered image is shown in Figure\ \ref{fig4}(b) right, where dark values
correspond to high-confidence estimates. As expected, pixels with a low
level of confidence are located in the neighborhood of image
discontinuities. For comparison we give the image denoised by the Non-Local
Means Filter with $21\times 21$ search windows and $9\times 9$ similarity
patches (PSNR$=31.51db$) and its square error, given in Figure\ \ref{fig4}
(c). The overall visual impression and the numerical results are improved
using our theory.

In Table \ref{tab psnr nlmeans}, we show a comparison of PSNR values of
Non-Local Means Filter computed with parameters propose in Buades et al \citep{buades2005review}
and with those proposed in our paper. It is easy to see that the visual
quality rises noticeably as the standard deviation $\sigma $ increases.
Nothing improves in the visual quality for $\sigma =10$, but it improves
with average $0.50db$ for $\sigma =20$ and average $0.98db$ for $\sigma =30$.
The comparison with several  filters is given in Table \ref{Table compar}. The  PSNR values show that the Non-Local Means Filter with proper parameters  is as good as more sophisticated methods, like  \citep{hirakawa2006image,kervrann2008local,hammond2008image,aharon2006rm}, and is better than  the filters proposed in  \citep{salmon2009nl,katkovnik2004directional,foinovel,roth2009fields,hirakawa2006image}.  The proposed approach  gives a denoising quality   which is competitive with that of the recent method BM3D \citep{dabov2007image}.

\begin{figure*}[tbp]
\begin{center}
\renewcommand{\arraystretch}{0.2} \addtolength{\tabcolsep}{-6pt} \vskip3mm {%
\fontsize{8pt}{\baselineskip}\selectfont
\begin{tabular}{cc}
\includegraphics[width=0.43\linewidth]{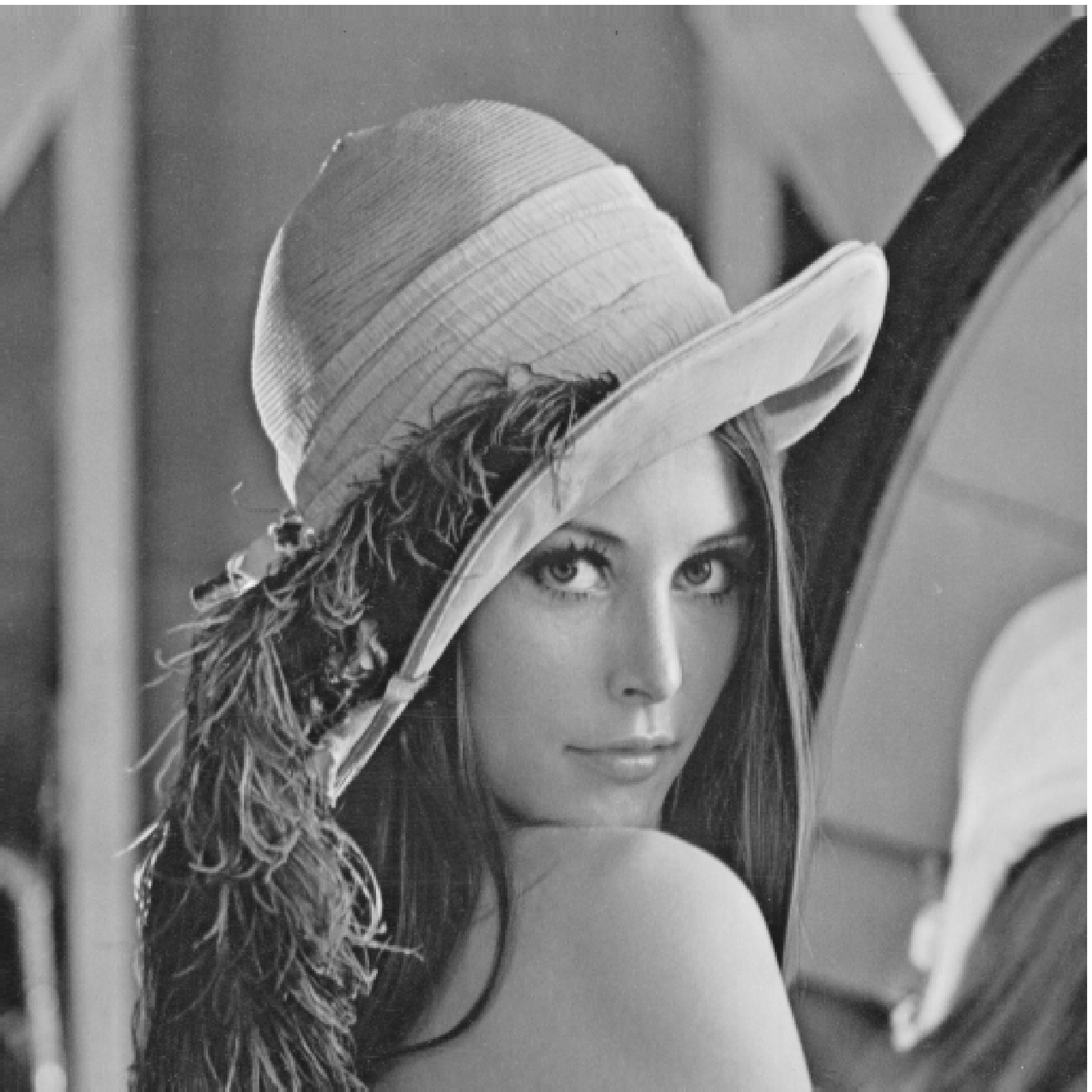} & %
\includegraphics[width=0.43\linewidth]{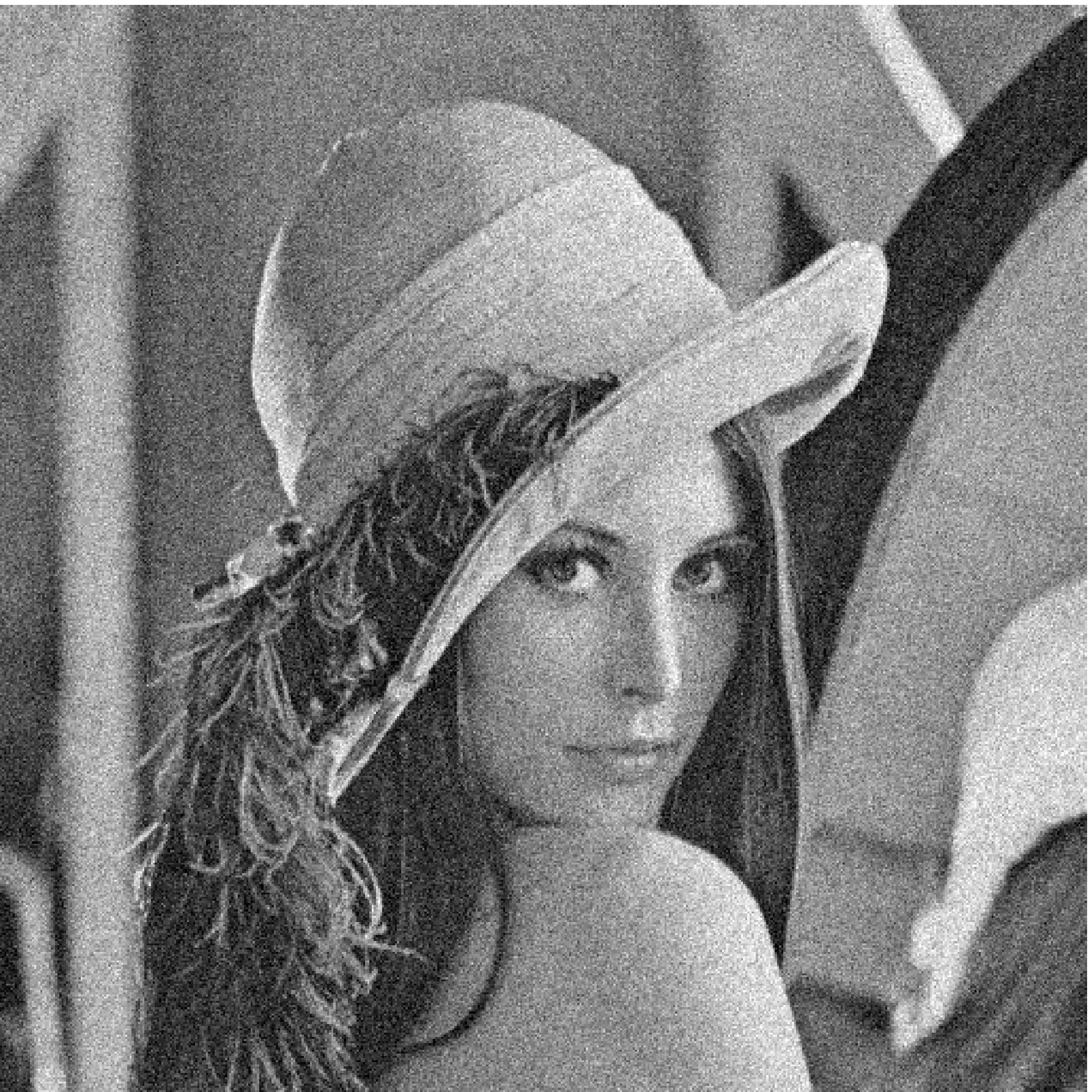} \\
\multicolumn{2}{c}{(a) Original $512\times 512$ image and noisy image with $%
\sigma =20$ ($22.11db$)} \\
\includegraphics[width=0.43\linewidth]{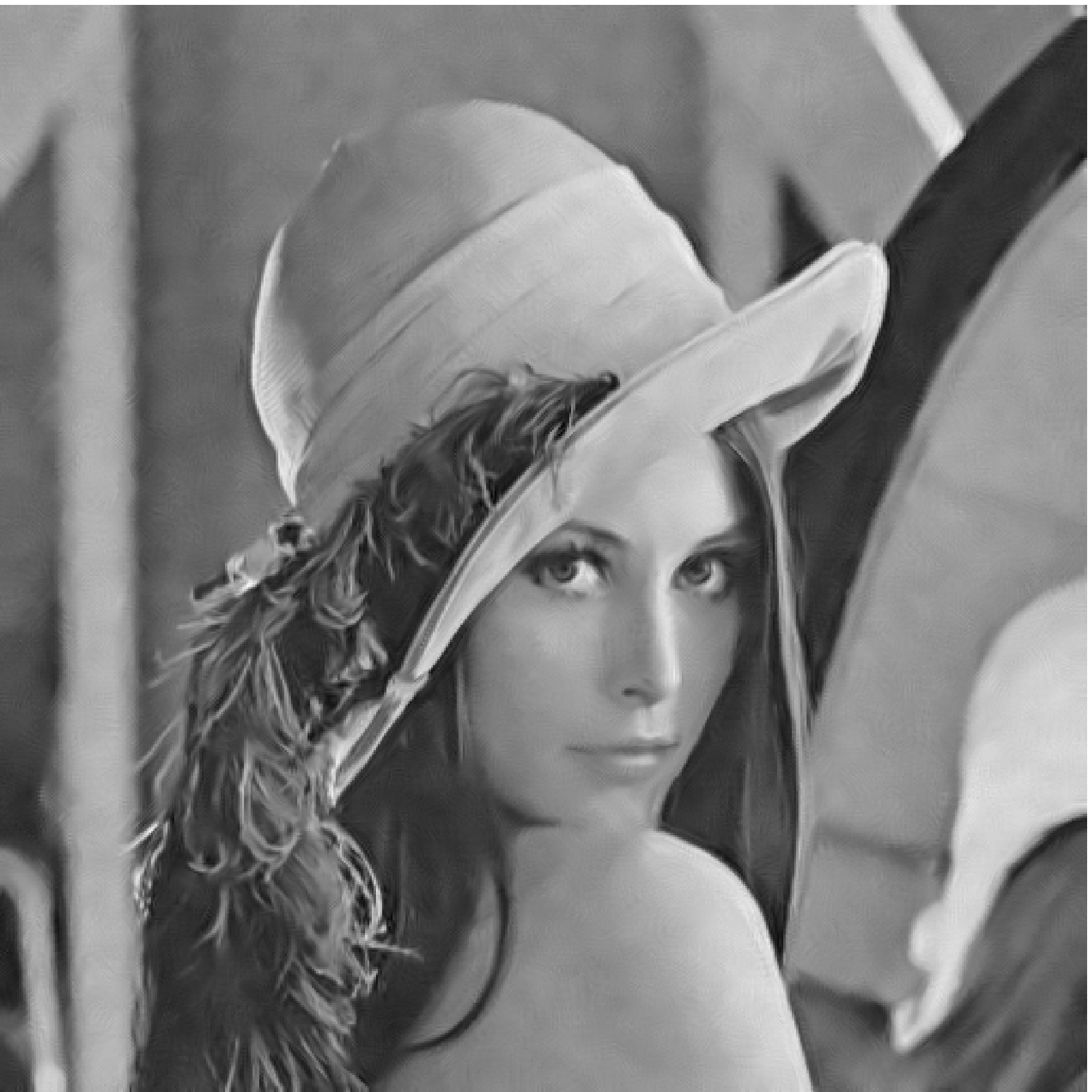} & %
\includegraphics[width=0.43\linewidth]{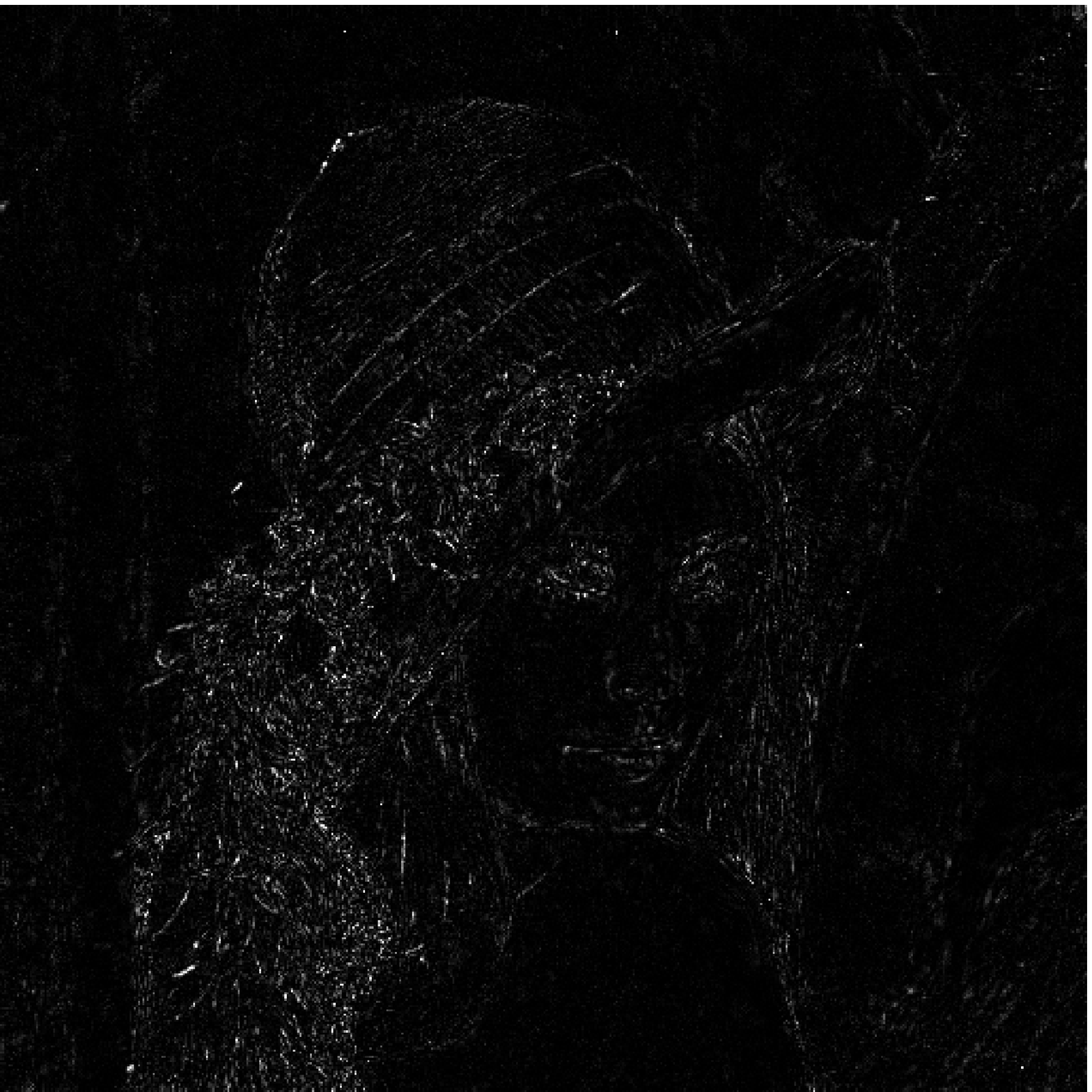} \\
\multicolumn{2}{c}{(b) Image denoised with our parameters ($32.39db$) and
its square error} \\
\includegraphics[width=0.43\linewidth]{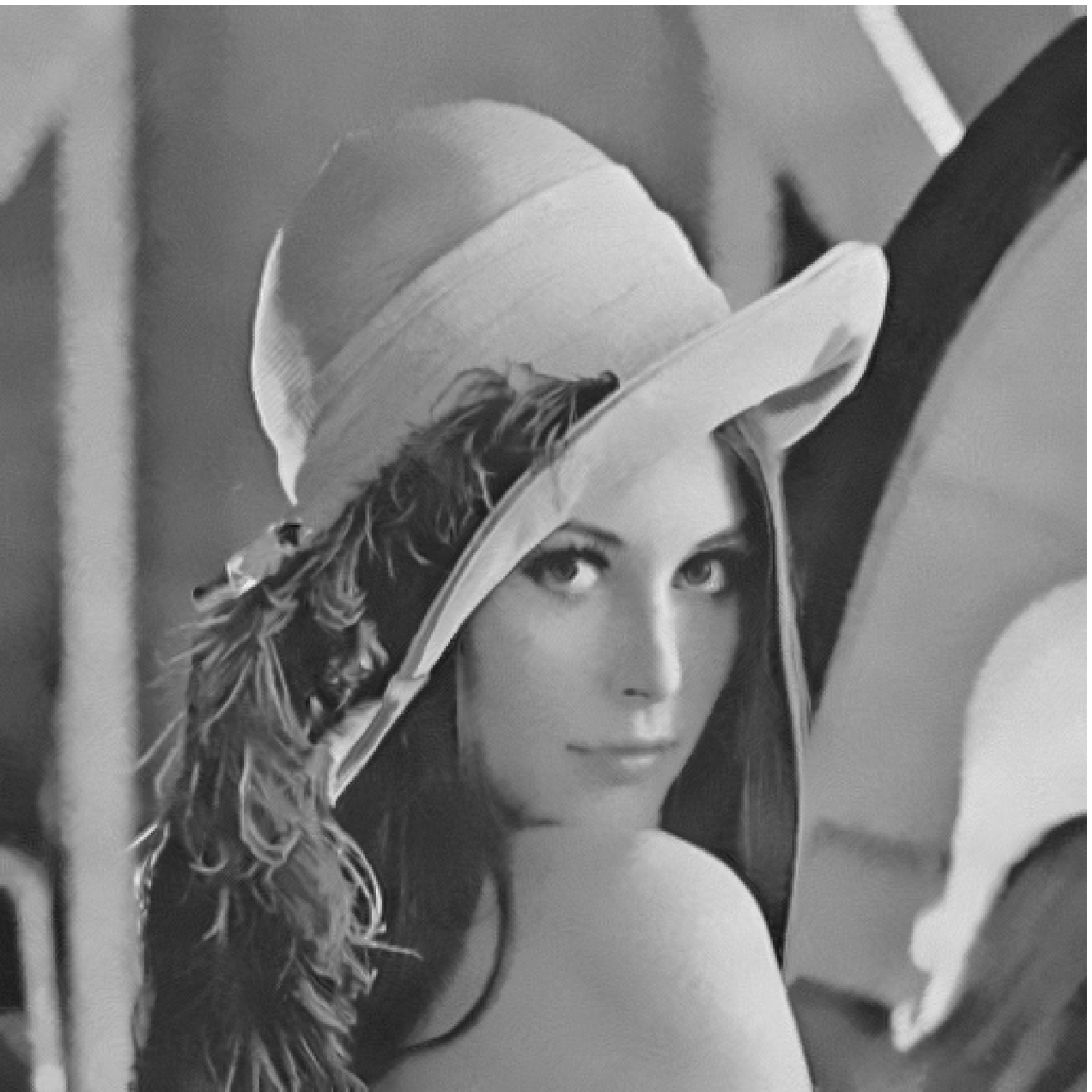} & %
\includegraphics[width=0.43\linewidth]{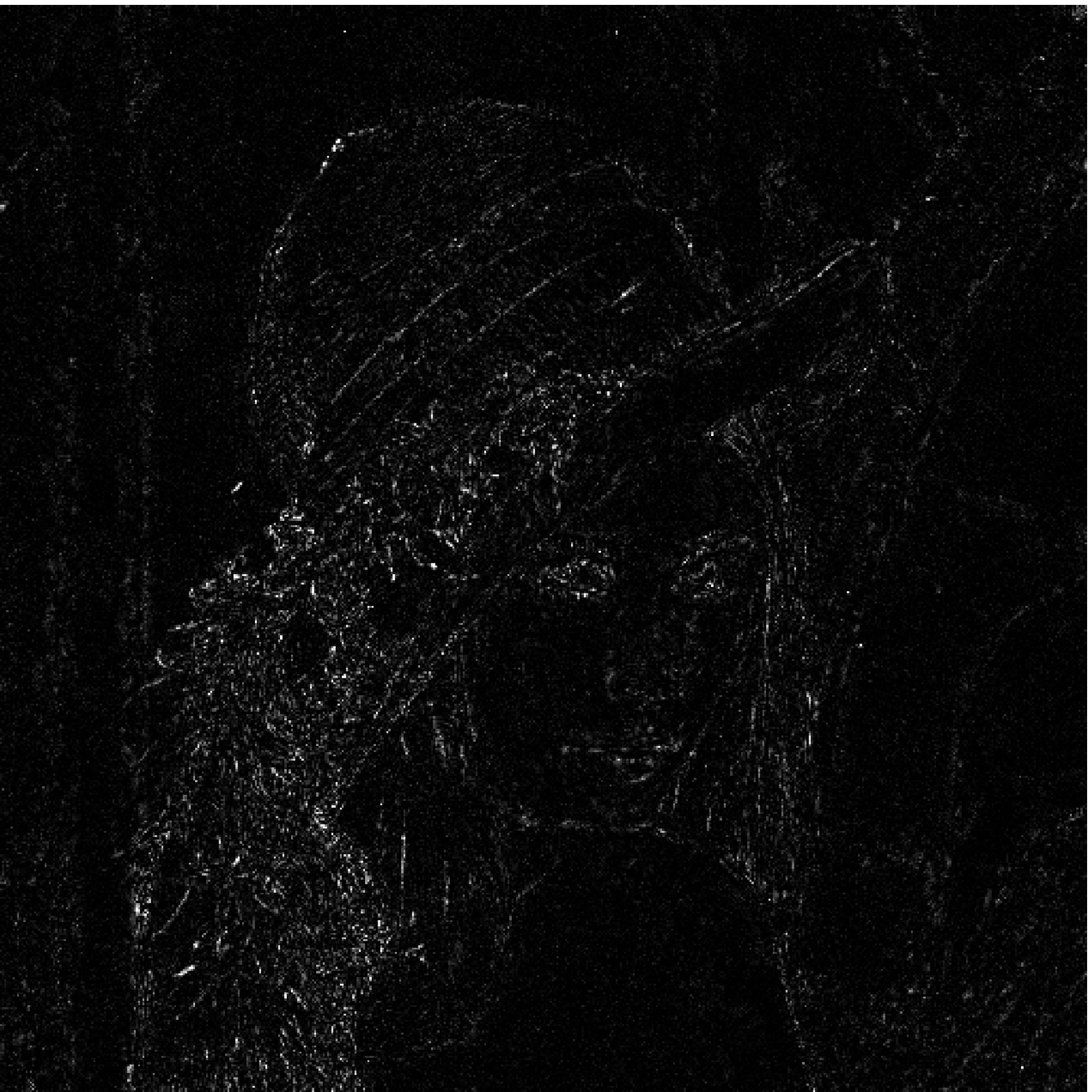} \\
\multicolumn{2}{c}{(c) Image denoised with Buade's parameters ($31.51db$)
and its square error}%
\end{tabular}
} \vskip1mm
\end{center}
\par
\rule{0pt}{-0.2pt}%
\caption{Results of denoising "Lena" $512\times 512$ image.}
\label{fig4}
\end{figure*}

\section{Conclusion \label{sec conclusion}}

We have proposed  new theorems of Non-Local Means Filter, based on optimization of parameters in the weighted means approach.
Our analysis shows that a small search window is preferred rather than the whole image and  a large similarity patch ($m=21\times 21$) is also preferred rather than the small similarity patch ($m=7\times 7$).
 The proposed theorems improve the usual parameters of  Non-Local Means Filter both numerically and visually in denoising performance.
We hope that the
convergence theorems for the Non-Local Means Filter that we deduced can also bring similar improvements for recently developed
algorithms where the basic idea of the Non-Local means filter is used.

\section{\label{sec proofs of mains results} Proofs of the main results}

\subsection{\label{proof of th oracle}Proof of Theorem \protect\ref%
{th_oracle}}

Denoting for brevity
\begin{eqnarray}
I_1
&=&
\left( \sum_{x\in \mathbf{I}}w^*_h(x)\rho_{f,x_0}(x)\right) ^{2}
\nonumber\\&=&
\left(\frac{%
\displaystyle \sum_{\|x-x_0\|_{\infty}\leq h} e^{-\frac{\rho_{f,x_0}^2(x)}{H^2} }{%
\rho}(x)}{\displaystyle \sum_{\|x-x_0\|_{\infty}\leq h}e^{-\frac{\rho_{f,x_0}^2(x)}{%
H^2} }}\right)^2,
\label{def I1}
\end{eqnarray}
and
\begin{eqnarray}
I_2
&=&
\sigma^2 \sum_{x\in \mathbf{I}}\left(w^*_h(x)\right)^2
\\&=&
 \sigma^2\frac{
\displaystyle \sum_{\|x-x_0\|_{\infty}\leq h} e^{-2\frac{\rho_{f,x_0}^2(x)}{H^2} }}{%
\left(\displaystyle \sum_{\|x-x_0\|_{\infty}\leq h}e^{-\frac{\rho_{f,x_0}^2(x)}{H^2}
}\right)^2},
\label{def I2}
\end{eqnarray}
then we have
\begin{equation}
g(w^*_h(w))=I_1+I_2.  \label{MSE bound}
\end{equation}
Noting that $te^{-\frac{t^2}{H^2}}$, $t\in [0, H/\sqrt 2)$ is increasing,
it is easy to see that
\begin{equation}
\begin{split}
&\sum_{\|x-x_0\|_{\infty}\leq h} e^{-\frac{
L^2\|x-x_0\|_{\infty}^{2\beta}}{H^2} }L\|x-x_0\|^{\beta}_{\infty}
\\&\leq
\sum_{\|x-x_0\|_{\infty}\leq h} L\|x-x_0\|^{\beta}_{\infty}e^{-\frac{
L^2\|x-x_0\|_{\infty}^{2\beta}}{H^2}}
\\ & \leq
\sum_{\|x-x_0\|_{\infty}\leq h}
L\|x-x_0\|^{\beta}_{\infty}\leq 4Lh^{\beta+2}n.
\end{split}
\label{inequ Lx}
\end{equation}

Since $e^{-\frac{t^2}{H^2}}$,
$t\in [0,H/\sqrt 2)$ is decreasing, Using one term Taylor expansion,
\begin{equation}
\begin{split}
&\sum_{\|x-x_0\|_{\infty}\leq h}e^{-\frac{L^2\|x-x_0\|_{
\infty}^{2\beta}}{H^2} }
\\&\geq
\sum_{\|x-x_0\|_{\infty}\leq h}e^{-\frac{L^2\|x-x_0\|_{\infty}^{2\beta}}{H^2}%
}
\\&\geq \sum_{\|x-x_0\|_{\infty}\leq h} \left(1-\frac{L^2\|x-x_0\|_{
\infty}^{2\beta}}{H^2}\right)\geq 2h^2n.
 \label{weight taylor}
 \end{split}
\end{equation}
The above three inequalities (\ref{def I1}), (\ref{weight taylor}) and (\ref{inequ Lx}) imply that
\begin{equation}
I_1 \leq 4L^2h^{2\beta}.  \label{I1 bound}
\end{equation}
Taking into account the inequality
\begin{equation*}
 \sum_{\|x-x_0\|_{\infty}\leq h} e^{-2\frac{\rho_{f,x_0}^2(x)}{H^2} }
 \leq  \sum_{\|x-x_0\|_{\infty}\leq h}1=4h^2n,
\end{equation*}
(\ref{def I2}) and (\ref{weight taylor}), it is easily seen that
\begin{equation}
I_2\leq \frac{\sigma^2}{h^2n}.  \label{I2 bound}
\end{equation}
Combining (\ref{MSE bound}), (\ref{I1 bound}), and (\ref{I2 bound}), we give
\begin{equation}
g(w_h^*)\leq 4L^2h^{2\beta}+\frac{\sigma^2}{h^2n}.
\label{upper bound ineq}
\end{equation}
Let $h$ minimize the latter term of the above inequality. Then
\begin{equation*}
8\beta L ^2 h^{2\beta-1}-\frac{2\sigma^2}{h^3n}=0
\end{equation*}
from which we infer that
\begin{equation}
h=\left(\frac{\sigma^2}{4\beta L^2}\right)^{\frac{1}{2\beta+2}}n^{-\frac{1}{%
2\beta+2}}.  \label{h value}
\end{equation}
Substituting (\ref{h value}) to (\ref{upper bound ineq}) leads to
\begin{equation*}
g(w_h^*)\leq \frac{2^{\frac{2\beta + 6}{2\beta +2}} \sigma^{\frac{4\beta}{%
2\beta +2}} L^{\frac{4}{2\beta +2}} }{\beta^{\frac{2\beta}{2\beta +2}}} n^{-%
\frac{2\beta}{2\beta+2}}.
\end{equation*}
Therefore (\ref{upper_bound}) implies (\ref{rate_oracle}).

\subsection{\label{proof of similar condition} Proof of Theorem \protect\ref%
{th similar function}}

First, we prove the following lemma:

\begin{lemma}
\label{Lemma s} Suppose that $S(x)$ is given by (\ref{defi sx}), then there are two constants $c_{2}$ and $c_{3}$, such that for
any $0\leq z\leq c_{2}m^{1/2},$
\begin{equation*}
\mathbb{P}\left( \left\vert S(x)\right\vert \geq z\sqrt{m}\right) \leq 2\exp \left(
-c_{3}z^{2}\right) .
\end{equation*}
\end{lemma}

\begin{proof}
Let $m$ is given by (\ref{defi m}).
Denote $\xi \left( y\right) =\zeta \left( y\right) ^{2}-2\sigma
^{2}+2\Delta_{x_0,x} \left( y\right) \zeta \left( y\right) .$ Since $\zeta
\left( y\right) $ is a normal random variable with mean $0$ and variance $%
2\sigma ^{2}$, there exist two positive constants $t_{0}$ and $c_{4}$
depending only on $\beta ,$ $L$, and $\sigma ^{2}$ such that $\phi
_{y}\left( t\right) =Ee^{t\xi (y)}\leq c_{4},$ for any $\left\vert
t\right\vert \leq t_{0}.$ Let $\psi _{y}(t)=\ln \phi _{y}\left( t\right) $
be the cumulant generating function. By Chebyshev's exponential inequality
we get
\begin{equation}
\mathbb{P}\{S(x)>z\sqrt{m}\}\leq \exp \left(-t\sqrt{m}z+\sum_{y\in \mathbf{U}_{x_{0},\eta }}\psi _{y}(t)\right),
\label{inquality sx}
\end{equation}
for any $\left\vert t\right\vert \leq t_{0}$ and for any $z>0.$ By three
term Taylor expansion, for $\left\vert t\right\vert \leq t_{0},$%
\begin{equation}
\psi _{y}(t)=\psi _{y}(0)+t\psi _{y}^{\prime }(0)+\frac{t^{2}}{2}\psi
_{y}^{\prime \prime }{(\theta t)},
\label{equ psi}
\end{equation}%
where $\left\vert \theta \right\vert \leq 1,$ $\psi _{y}(0)=0,$ $\psi
_{y}^{\prime }(0)=\mathbb{E}\xi (y)=0$ and%
\begin{equation*}
0\leq \psi _{y}^{\prime \prime }(t)=\frac{\phi _{y}^{\prime \prime }\left(
t\right) \phi _{y}\left( t\right) -\left( \phi _{y}^{\prime }\left( t\right)
\right) ^{2}}{\left( \phi _{y}\left( t\right) \right) ^{2}}\leq \frac{\phi
_{y}^{\prime \prime }\left( t\right) }{\phi _{y}\left( t\right) }.
\end{equation*}%
Since, by Jensen's inequality $\mathbb{E}e^{t\xi (y)}\geq e^{t\mathbb{E}\xi (y)}=1,$ we arrive
at the following upper bound%
\begin{equation*}
\psi _{y}^{\prime \prime }(t)\leq \phi _{y}^{\prime \prime }\left( t\right)
=\mathbb{E}\left(\xi ^{2}(y)e^{t\xi (y)}\right).
\end{equation*}%
Using the elementary inequality $x^{2}e^{x}\leq e^{3x},$ $x\geq 0,$ we have,
for $\left\vert t\right\vert \leq t_{0}/3,$
\begin{eqnarray}
\psi _{y}^{\prime \prime }(t)
&\leq &
\frac{9}{t_{0}^{2}}\mathbb{E}\left(\left( \frac{t_{0}}{3}%
\xi (y)\right) ^{2}e^{\frac{t_{0}}{3}\xi (y)}\right)
\nonumber\\&\leq&
 \frac{9}{t_{0}^{2}}%
\mathbb{E}e^{t_{0}\xi (y)}\leq \frac{9}{t_{0}^{2}}c_{4}.
\label{ineq psi}
\end{eqnarray}
The inequality (\ref{ineq psi}) combining with (\ref{equ psi}) implies that for $\left\vert t\right\vert \leq t_{0},$%
\begin{equation*}
0\leq \psi _{y}(t)\leq \frac{9c_{4}}{2t_{0}^{2}}t^{2}.
\end{equation*}%
Then (\ref{inquality sx}) becomes
\begin{equation}
\mathbb{P}\left( S(x)>z\sqrt{m}\right) \leq \exp \left(-tz\sqrt{m}+\frac{9c_{4}}{%
2t_{0}^{2}}mt^{2}\right).  \notag
\end{equation}%
If $t=c^{\prime }z/\sqrt{m}\leq t_{0}/3$, we obtain%
\begin{equation}
\mathbb{P}\left( S(x)>z\sqrt{m}\right) \leq \exp \left( -c'z^{ 2}\left( 1-\frac{%
9c_{4}}{2t_{0}^{2}}c^{\prime }\right) \right) .  \notag
\end{equation}%
Choosing $c^{\prime }$ sufficiently small we arrive at%
\begin{equation*}
\mathbb{P}\left( S(x)>z\sqrt{m}\right) \leq \exp \left( -c_{3}z^{2}\right) ,
\end{equation*}%
for some constant $c_{3}>0.$ In the same way we show that%
\begin{equation*}
\mathbb{P}\left( S(x)<-z\sqrt{m}\right) \leq \exp \left( -c_{3}z^{2}\right) .
\end{equation*}%
This proves the lemma.
\end{proof}

Finally, we turn to the proof of Theorem \ref{th similar function}. Applying
Lemma \ref{Lemma s} with $z=\sqrt{\frac{1}{c_3}\ln n^2}$, we see that
\begin{equation*}
\mathbb{P}\left( \frac{1}{m}\left\vert S(x)\right\vert \geq \frac{\sqrt{\frac{1}{c_{3}
}\ln n^2}}{\sqrt{m}}\right) \leq 2\exp \left( -\ln n^2\right)=\frac{2}{n^2}.
\end{equation*}
From this inequality we easily deduce that
\begin{eqnarray*}
&&\mathbb{P}\left( \max_{x\in \mathbf{U}_{x_{0},h}^{\prime }}\frac{1}{m}\left\vert
S(x)\right\vert \geq \frac{\sqrt{\frac{1}{c_{3}}\ln n^2}}{\sqrt{m}}\right)
\\&\leq&
 \sum_{x\in \mathbf{U}_{x_{0},h}^{\prime }}\mathbb{P}\left( \frac{1}{m}\left\vert
S(x)\right\vert \geq \frac{\sqrt{\frac{1}{c_{3}}\ln n^2}}{\sqrt{m}}%
\right)\leq \frac{2}{n}.
\end{eqnarray*}
Taking $m=(2N\eta+1)^2=c_0^2n^{1-2\alpha}$, we arrive at
\begin{equation}  \label{ineq p s}
\mathbb{P}\left( \max_{x\in \mathbf{U}_{x_0,h} }\frac{1}{m}
\left|S(x)\right|\geq c_1n^{\alpha-\frac{1}{2}}\right) \leq
O\left(n^{-1}\right).
\end{equation}

The local H\"{o}lder condition (\ref{Local Lip cond}) implies that
\begin{equation}
\left|\widehat{\rho}_{x_0}^{2}(x)-\rho^2(x)\right|\leq O\left( n^{-\frac{2\beta}{%
2\beta+2}}\right) + \frac{1}{m}\left|S(x)\right|.
 \label{ineq distance rho}
\end{equation}
Combining (\ref{ineq p s}) and (\ref{ineq distance rho}), we get
\begin{eqnarray}
&&\mathbb{P}\left( \max_{x\in \mathbf{U}_{x_0,h}}\left|\widehat{\rho}_{x_0}%
^{2}(x)-\rho^2(x)\right|\geq O\left(n^{-\frac{2\beta}{2\beta+2}}\right)+
c_1n^{\alpha-\frac{1}{2}}\right)
\nonumber\\&\leq&
O\left(n^{-1}\right).
\label{equation p}
\end{eqnarray}
Because the  condition $\frac{(1-\beta)^+}{2\beta+2}<\alpha<\frac{1}{2}$ implies that
\begin{equation}
n^{\alpha-\frac{1}{2}}> n^{-\frac{2\beta}{2\beta+2}},
\end{equation}
the
 inequality  (\ref{rate similar function}) holds.

\subsection{\label{proof of rate estimator}Proof of Theorem \protect\ref{th
rate estimator}}

Taking into account (\ref{estimator similar}), (\ref{estimate}), and the
independence of $\epsilon (x)$, we have
\begin{equation}
\mathbb{E}\left(|\widehat{f}'_{h }(x_{0})-f(x_{0})|^{2}\big|Y(x),x\in
\mathbf{I}''_{x_{0}}\right)
\leq g^{\prime }(\widehat{w}_h),  \label{MSE estimator}
\end{equation}%
where
\begin{equation*}
g^{\prime }(w)=\left( \sum_{x\in \mathbf{U}'_{x_{0},h}}{w}
(x)\rho_{f,x_0} (x)\right) ^{2}+\sigma ^{2}\sum_{x\in \mathbf{U}'_{x_{0},h}}%
{w}^{2}(x).
\end{equation*}%
From the proof of Theorem \ref{th_oracle}, we infer that
\begin{equation}
g^{\prime }(w^{\ast }_{h})\leq \frac{3}{2}\left( \frac{2^{\frac{2\beta +6}{%
2\beta +2}}\sigma ^{\frac{4\beta }{2\beta +2}}L^{\frac{4}{2\beta +2}}}{\beta
^{\frac{2\beta }{2\beta +2}}}n^{-\frac{2\beta }{2\beta +2}}\right) .
\label{gx bound}
\end{equation}%

By Theorem \ref{th similar function} and its proof, for $\widehat{\rho}'_{x_0}$
 defined by (\ref{estimator similar}), there is a constant $c_{1}$
such that
\begin{eqnarray}
&&\mathbb{P}\left( \max_{x\in \mathbf{U}'_{x_{0},h}}\left\vert \widehat{\rho }_{x_0}
^{'2}(x)-\rho_{f,x_0} ^{2}(x)\right\vert \geq c_{1}n^{\alpha -\frac{1}{2}}\sqrt{\ln n}%
\right)
\nonumber\\&=&
O\left( n^{-1}\right) .  \label{rate01 similar function}
\end{eqnarray}
Let $\mathbf{B}=\left( \max_{x\in \mathbf{U}'_{x_{0},h}}\left\vert \widehat{%
\rho }_{x_0}^{'2}(x)-\rho_{f,x_0} ^{2}(x)\right\vert \leq c_{1}n^{\alpha -\frac{1}{2}%
}\right) $. On the set $\mathbf{B}$, we have $\rho_{f,x_0} ^{2}(x)-c_{1}n^{\alpha -%
\frac{1}{2}}<\widehat{\rho }_{x_0}^{'2}(x)<\rho_{f,x_0} ^{2}(x)+c_{1}n^{\alpha -\frac{1}{2}}
$, from which we infer that
\begin{equation*}
\begin{split}
\widehat{w}_{h}(x)
& =
\frac{e^{-\frac{\widehat{\rho }_{x_0}^{'2}(x)}{H^{2}}}}{\sum_{y\in
\mathbf{U}_{x_{0},h}^{\prime }}e^{-\frac{\widehat{\rho }_{x_0}^{'2}(y)}{H^{2}}}}%
\\&\leq
\frac{e^{-\frac{{\rho }_{f,x_0}^{2}(x)-c_{1}n^{\alpha -\frac{1}{2}}}{H^{2}}}}{%
\sum_{y\in \mathbf{U}_{x_{0},h}^{\prime }}e^{-\frac{{\rho }_{f,x_0}
^{2}(y)+c_{1}n^{\alpha -\frac{1}{2}}}{H^{2}}}} \\
& \leq
\frac{e^{-\frac{{\rho }_{f,x_0}^{2}(x)}{H^{2}}}\left( 1+2\frac{c_{1}n^{\alpha
-\frac{1}{2}}}{H^{2}}\right) }{\sum_{y\in \mathbf{U}_{x_{0},h}^{\prime }}e^{-%
\frac{{\rho }_{f,x_0}^{2}(y)}{H^{2}}}\left( 1-\frac{c_{1}n^{\alpha -\frac{1}{2}}}{%
H^{2}}\right) }
\\&=
\left( \frac{1+\frac{2c_{1}n^{\alpha -\frac{1}{2}}}{H^{2}}}{%
1-\frac{c_{1}n^{\alpha -\frac{1}{2}}}{H^{2}}}\right) w^{\ast }_{h}
\end{split}%
\end{equation*}%
This implies that
\begin{equation*}
g^{\prime }(\widehat{w}_{h})\leq \left( \frac{1+\frac{2c_{1}n^{\alpha -\frac{%
1}{2}}}{H^{2}}}{1-\frac{c_{1}n^{\alpha -\frac{1}{2}}}{H^{2}}}\right)
^{2}g^{\prime  }(w^*_{h}).
\end{equation*}%
Consequently, the inequality (\ref{MSE estimator}) becomes
\begin{eqnarray}
&&\mathbb{E}\left(|\widehat{f}'_{h }(x_{0})-f(x_{0})|^{2}\big|Y(x),x\in \mathbf{I}''_{x_{0}},%
\mathbf{B}\right)
\nonumber\\&\leq &
\left( \frac{1+\frac{2c_{1}n^{\alpha -\frac{1}{2}}}{H^{2}}}{%
1-\frac{c_{1}n^{\alpha -\frac{1}{2}}}{H^{2}}}\right) ^{2}g^{\prime  }(w^*_{h}).  \label{expection bound}
\end{eqnarray}
Since the function $f$ satisfies the local H\"{o}lder condition (\ref{Local Lip cond}),
\begin{equation}
\mathbb{E}\left( |\widehat{f}'_{h}(x_{0})-f(x_{0})|^{2}\big|Y(x),x\in \mathbf{I}''_{x_{0}}\right) <g'(\widehat{w}_h)\leq c_{2},  \label{expection bound K}
\end{equation}%
for a constant $c_{2}>0$ depending only on $\beta $, $L$, and $\sigma $.
Combining (\ref{rate similar function}), (\ref{expection bound}), and (\ref%
{expection bound K}), we have
\begin{equation*}
\begin{split}
\mathbb{E}& \left( |\widehat{f}'_{h}(x_{0})-f(x_{0})|^{2}\big|Y(x),x\in \mathbf{I}''_{x_{0}},\right)  \\
=& \mathbb{E}\left(|\widehat{f}'_{h }(x_{0})-f(x_{0})|^{2}\big|Y(x),x\in \mathbf{I}''_{x_{0}},%
\mathbf{B}\right)\mathbb{P}(\mathbf{B}) \\
& +\mathbb{E}\left(|\widehat{f}'_{h }(x_{0})-f(x_{0})|^{2}\big|Y(x),x\in \mathbf{I}''_{x_{0}},%
\overline{\mathbf{B}}\right)\mathbb{P}(\overline{\mathbf{B}}) \\
\leq & \left( \frac{1+\frac{2c_{1}n^{\alpha -\frac{1}{2}}}{H^{2}}}{1-\frac{%
c_{1}n^{\alpha -\frac{1}{2}}}{H^{2}}}\right) ^{2}g'(w^*_h)+O\left( n^{-1}\right)
\end{split}%
\end{equation*}%
%
Now, the assertion of the theorem is obtained easily if we note the inequality
(\ref{gx bound}).


\end{document}